\documentclass[a4paper,USenglish, 11pt]{article}

\usepackage{graphics}
\usepackage{latexsym,amssymb,amsmath}
\usepackage{multirow}
\usepackage{fullpage}

\usepackage{authblk}

\newcommand{\itemshort}{
\setlength{\itemsep}{0mm}
\setlength{\parskip}{0mm}
}

\newtheorem{algorithm}{Algorithm}
\newtheorem{definition}{Definition}
\newtheorem{lemma}{Lemma}
\newtheorem{theorem}{Theorem}
\newtheorem{corollary}{Corollary}
\newenvironment{proof}{\noindent {\bf Proof: }}{\hfill$\square$}
\newcommand{\qed}{\hfill$\square$}

\title{Self-Stabilizing Maximal Matching and Anonymous Networks}
\date{}

\author[1]{Johanne Cohen}
\author[2]{Jonas Lef\`evre}
\author[3]{Khaled Ma\^amra\thanks{This work was partially funding by DIGITEO, project RI2A\^{}2}}
\author[3]{Laurence Pilard}
\author[3]{Devan Sohier}
\affil[1]{LRI, CNRS, Universit\'e Paris Sud, Universit\'e Paris-Saclay, France. 
  \texttt{johanne.cohen@lri.fr}}
\affil[2]{LIX, CNRS, \'Ecole Polytechnique, Universit\'e Paris-Saclay, France. \texttt{jlefevre@lix.polytechnique.fr}}
\affil[3]{LI-PaRAD, Universit\'e Versailles-St. Quentin, Universit\'e Paris-Saclay, France. 
  \texttt{\{khaled.maamra, laurence.pilard, devan.sohier\}@uvsq.fr}}

\begin{document}
\maketitle

\begin{abstract}
 We propose a self-stabilizing algorithm for computing a maximal matching in an anonymous network. The complexity 
is  $O(n^3)$ moves with high probability, under the adversarial distributed daemon.  In this algorithm, each node can determine whether one of its neighbors points to it or to another node, leading to a contradiction with the anonymous assumption. To solve this problem, we provide under the classical link-register model, a self-stabilizing algorithm that gives a unique name to a link such that this name is shared by both extremities of the link.
\end{abstract}


\section{Introduction}

Matching problems have received a lot of attention in different areas. Dynamic load balancing and job scheduling in parallel and distributed networks can be solved by algorithms using a matching set of communication links~\cite{BerenbrinkFM08,GhoshM96}. Moreover, the matching  problem has been recently  studied in the algorithmic game theory.  Indeed, the seminal problem relative to matching introduced by Knuth  is the stable marriage problem \cite{Knuth}. This problem can be modeled as a game with economic interactions such as two-sided markets \cite{AckermannGMRV11} or as a game with preference relations in a social network \cite{Hoefer13}.  But, all distributed algorithms proposed in the game theory domain use identities while we are interested in anonymous networks, \emph{i.e.} without identity.

In graph theory, a \emph{matching} $M$ in a graph is a set of edges without common vertices.  A matching is \emph{maximal} if no proper superset of $M$ is also a matching. A \emph{maximum} matching is a maximal matching with the highest cardinality among all possible maximal matchings. In this paper, we present a self-stabilizing algorithm for finding a maximal matching. Self-stabilizing algorithms \cite{Dijkstra74,Dolev00}, are distributed algorithms that recover after any transient failure without external intervention \emph{i.e.} starting from any arbitrary initial state, the system eventually converges to a correct behavior. The environment of self-stabilizing algorithms is modeled by the notion of \emph{daemon}. A daemon allows to capture the different behaviors of such algorithms accordingly to the execution environment. Two major types of daemons exist: the \emph{sequential} and the \emph{distributed} ones. The sequential daemon means that exactly one eligible process is scheduled for execution at a time. The distributed daemon means that any subset of eligible processes is scheduled for execution at a time. In an orthogonal way, a daemon can be \emph{fair} (meaning that every eligible process is eventually scheduled for execution) or \emph{adversarial} (meaning that the daemon only guarantees global progress, \emph{i.e.} at any time, at least one eligible process is scheduled for execution). 

In this paper we provide two self-stabilizing algorithms. The first one, called the \emph{matching algorithm}, is a randomized algorithm for finding a maximal matching in an anonymous network.  We show the algorithm stabilizes in expected $O(n^3)$ moves under the adversarial distributed daemon. This is the first algorithm solving this problem assuming at the same time an anonymous network and an adversarial distributed daemon. 
In this algorithm, nodes have pointers and a node can determine whether its neighbor points to it or to another node, leading to a contradiction with the anonymous assumption. Indeed, to know which node a neighbor is pointing to, the usual way is to use identities. To solve this problem, and this is the first paper giving a solution, we provide a self-stabilizing algorithm where we assume the classical link-register model, \emph{i.e.} a node communicates with a neighbor through a register associated to the link from the node to this neighbor. This algorithm, called the \emph{link-name algorithm}, gives names to communication links such that (i) both nodes at the extremity of a link $\ell$ know the name of $\ell$ and (ii) a node cannot have two distinct incident links with the same name. At the end of this paper, we will see how to rewrite the matching algorithm using output of the link-name algorithm, allowing then a process to know if one of its neighbors points to it without using identity.

\section{Related Works}

Several self-stabilizing algorithms have been proposed to compute maximal matching in unweighted or weighted general graphs.  For an unweighted graph, Hsu and Huang \cite{HsuH92} gave the first algorithm and proved a bound of $O(n^3)$ on the number of moves under a sequential adversarial daemon. The complexity analysis is completed by Hedetniemi et al. \cite{HedetniemiJS01} to $O(m)$ moves. Manne et al. \cite{ManneMPT11} presented a self-stabilizing algorithm for finding a $2/3$-approximation of a maximum matching.  The complexity of this algorithm  is proved to be $O(2^n)$ moves under a distributed adversarial daemon. 
In a weighted graph, Manne and Mjelde \cite{ManneM07} presented the first self-stabilizing algorithm for computing a weighted matching of a graph with an $1/2$-approximation to the optimal solution. They that established their algorithm stabilizes after at most exponential number of moves under any adversarial daemon (\emph{i.e.} sequential or distributed). Turau and Hauck \cite{TurauH11a} gave a modified version of the previous algorithm that stabilizes after $O(nm)$ moves under any adversarial daemon.

All algorithms presented above, but the Hsu and Huang \cite{HsuH92}, assume nodes have unique identity. The Hsu and Huang's algorithm is the first one working in an anonymous network. This algorithm operates under any sequential daemon (fair or adversarial) in order to achieve symmetry breaking. Indeed, Manne et al. \cite{ManneMPT09} proved   that in some anonymous networks there exists no deterministic self-stabilizing solution to the maximal matching problem under a synchronous daemon. This is a general result that holds under either the fair or the adversarial distributed daemon.  This also holds whatever the communication and atomicity model (the state model with guarded rule atomicity or the link-register model with read/write atomicity).     
Goddard et al. \cite{GoddardHJS08} proposed a generalized scheme that can convert any anonymous and deterministic algorithm that stabilizes under an adversarial sequential daemon into a randomized one that stabilizes under a distributed daemon, using only constant extra space and without identity. The expected slowdown is  bounded by $O(n^3)$ moves.  The composition of these two algorithms  can compute a maximal matching in $O(mn^3)$ moves in  an anonymous network  under a distributed daemon.

In anonymous networks, Gradinariu and Johnen \cite{GradinariuJ01} proposed a self-stabilizing probabilistic algorithm to give processes a local identity that is unique within distance $2$. They used this algorithm to run the Hsu and Huang's algorithm under an adversarial distributed daemon.  However, only a finite stabilization time was proved. Chattopadhyay et al. \cite{Chattopadhyay} improved this result by giving a maximal matching algorithm with $O(n)$ expected rounds  complexity under the   fair distributed daemon.  Note that a \emph{round} is a minimal sequence of moves where each node makes at least one move.  It is straightforward to show that this algorithm stabilizes in $\Omega(n^2)$ moves, but Chattopadhyay et al. do not give any upper bound on the move complexity. 

The previous algorithm as well as the maximal matching algorithm presented in this work both assume an anonymous network and a distributed daemon. However the first algorithm assumes the   fair daemon while the second one does not make any fairness assumption. Moreover, no move complexity is given for the first algorithm while we will prove the second one converges in expected  $O(n^3)$ moves.  

The following table compares features of the aforementioned algorithms and ours. Among all adversarial distributed daemons and with the anonymous assumption, our algorithm provides the best complexity.
\setlength{\tabcolsep}{3.25pt}
\begin{table}[h]
  \centering
  \begin{tabular}[h]{|c|c|c|c|c|c|}\hline
                    &\multirow{2}{*}{\cite{HsuH92,HedetniemiJS01} } &  Composition &\multirow{2}{*}{\cite{GradinariuJ01}}  & \multirow{2}{*}{\cite{Chattopadhyay}}   & \multirow{2}{*}{This paper} \\
                    &   & \cite{HsuH92,HedetniemiJS01}   with  \cite{GoddardHJS08}    &                                                          &                                                             & \\\hline
 \multirow{2}{*}{Daemon}             &  adversarial    & adversarial    & adversarial    & fair   & adversarial  \\        
                                                     &    sequential    & distributed    &   distributed  &   distributed &   distributed            \\  \hline
\multirow{2}{*}{Complexity}        & \multirow{2}{*}{ $O(m)$  moves} & $O(mn^3)$         & \multirow{2}{*}{finite}        &       $O(n^3)$              & $O(n^3)$ moves     \\
  &  &    expected moves      &        &      expected moves     &  with high probability\\\hline
 \end{tabular}
\end{table}

When dealing with matching under anonymous networks, we have to overcome the difficulty that a process has to know if one of its neighbors points to it. In Hsu and Huang's paper \cite{HsuH92}, this difficulty is not even mentioned and the assumption a node can know if one of its neighbors points to it is implicitly made. However, this difficulty is mentioned in the Goddard et al. paper \cite{GoddardHS06}, where authors present an anonymous self-stabilizing algorithm for finding a $1$-maximal matching in trees and rings. 
To overcome this difficulty, authors assume that every two adjacent nodes share a private register containing an incorruptible link's number. Note that this problem does not appear for the vertex cover problem \cite{TurauH11} or the independent set problem \cite{ShiGH04} even in anonymous networks (see \cite{GuellatiK10} for a survey). Indeed, in these kind of problems, we do not try to build a set of edges, but a set of nodes.  So, a node does not point to anybody and it simply has to know whether or not one of its neighbors belongs to the set. In this paper, we propose a self-stabilizing solution for this problem without assuming any incorruptible memory.

\section{Model}

A system consists of a set of processes where two adjacent processes can communicate with each other. The communication relation is typically represented by a graph G = (V, E) where $|V | = n$ and $|E| = m$. Each process corresponds to a node in $V$ and two processes $u$ and $v$ are adjacent if and only if $(u,v)\in E$. The set of neighbors of a process $u$ is denoted by $N(u)$ and is the set of all processes adjacent to $u$. We assume an \emph{anonymous} system meaning that processes have no identifiers. Thus, two different processes having the same number of neighbors are undistinguishable. 

We distinguish  two communication models : the \emph{state model} and the \emph{link-register model}. We are going to define the state model, then we will define the link-register model by pointing out the differences with the state model. 

In the \emph{state model}, each process maintains a set of \emph{local variables} that makes up the \emph{local state} of the process.  A process can read its local variables and the local variables of  its neighbors, but it can write only in its own local variables. A \emph{configuration} $C$ is a set of the local states of all processes in the system. 
Each process executes the same algorithm that consists of a set of \emph{rules}. Each rule is of the form of $<guard> \to <command>$. The \emph{guard} is a boolean function over the variables of both the process and its neighbors. The \emph{command} is a sequence of actions assigning new values to the local variables of the process. 

A rule is \emph{enabled} in a configuration $C$ if the guard is true in $C$. A process is \emph{activable} in a configuration $C$ if at least one of its rules is enabled. 
An \emph{execution} is an alternate sequence of configurations and transitions ${\cal E} = C_0,A_0,\ldots,C_i,A_i, \ldots$, such that $\forall i\in  \mathbb{N}^*$, 
$C_{i+1}$ is obtained by executing the command of at least one rule 
that is enabled in $C_i$ (a process that executes such a rule makes a \emph{move}). More precisely, $A_i$ is the non empty set of enabled rules in $C_i$ that has been executed to reach $C_{i+1}$ such that each process has at most one of its rules in $A_i$. 
An \emph{atomic operation} is such that no change can takes place during its run, we usually assume an atomic operation is instantaneous. In the case of the state model, such an operation corresponds to a rule.
We use the following notation : $C_i \to C_{i+1}$. 
An execution is \emph{maximal} if it is infinite, or it is finite   and no process is activable in the last configuration. All algorithm executions considered in this paper are assumed to be maximal. 

A \emph{daemon} is a predicate on the executions. We consider only the most powerful one: the \emph{distributed daemon} that allows all executions described in the previous paragraph. 

An algorithm is \emph{self-stabilizing} for a given specification, if there exists a sub-set $\cal L$ of the set of all configurations such that : every execution starting from a configuration of $\cal L$ verifies the specification  (\emph{correctness})  and  starting from any configuration, every execution reaches a configuration of $\cal L$   (\emph{convergence}). $\cal L$ is called the set of \emph{legitimate configurations}. A \emph{probabilistic self-stabilizing} algorithm ensures (deterministic) correctness, but only ensures probabilistic convergence. 

A configuration is \emph{stable} if no process is activable in the configuration. 
Both algorithms presented here, are \emph{silent}, meaning that once the algorithm stabilized, no process is activable. In other words, all executions of a silent algorithm are finite and end in a stable configuration.
Note the difference with a non silent self-stabilizing algorithm that has at least one infinite execution with a suffix only containing legitimate configurations, but not stable ones. 

In the \emph{link-register model} \cite{DIM93}, each process maintains a set of \emph{registers} associated to each of its communication links. A process can read and write in its own registers, but it can only read the registers of its neighbors that are associated to one of its links. More formally, if $u$ is a process, $\forall v \in N(u)$, $u$ maintains a set of registers $reg_{uv}$. These registers belong to $u$ and are associated to the link $(u,v)$, thus $u$ can read and write in these registers and $v$ can read them. The set of all registers of a process is the \emph{local state} of the process. In this paper, we use the {link-register model}. In this model, we assume local unique name on edges/ports, classically named the \emph{port numbering model} in message-passing systems, as a syntactic tool to designate a specific register as well as its counterpart on the other side of the communication link.

The configuration and rule definitions remain the same, but the difference is in the execution definition. In the link-register model, usually, the minimal atomicity is not the rule, but the action (remember that a command is a set of actions). This atomicity is called the \emph{read/write atomicity}.
We define two types of actions: (i) the \emph{internal action} that is an action over internal variables, such that \texttt{i++}, and (ii) the \emph{communication action} that is an action of reading or writing in a register. An \emph{atomic action} here is the execution of a finite sequence of internal actions ended by one communication action. Then, a transition in an execution is a non-empty set of atomic actions such that each process has at most one of its atomic actions in the transition.

\section{Maximal matching algorithm}\label{sec:matching}

The \emph{matching algorithm} $\mathcal{A}_1$ presented in this section uses the state model given in the previous section and is based on the maximal matching algorithm given by Manne et al.~\cite{ManneMPT09}.
In algorithm $\mathcal{A}_1$, every node $u$ has one local variable $\beta_{u}$ representing the node $u$ is matched with. If $u$ is not matched, then $\beta_{u}$ is equal to $\bot$.
Algorithm $\mathcal{A}_1$ ensures that a maximal matching is eventually built. Formally, we require the following specification $\mathcal{M}$ for $\mathcal{A}_1$:

\begin{definition}[Specification $\mathcal{M}$]
For a graph $G=(V,E)$, the set $M = \{ (u,v): \beta_{u} {=} v \land \beta_{v} {=} u \}$ is a maximal matching of $G$, \emph{i.e.} $\mathcal{M}$ = $\mathcal{M}_1 \land \mathcal{M}_2 \land \mathcal{M}_3$ holds:
\begin{itemize}
\item[]$\mathcal{M}_1$: ~ $\forall (u,v) \in M: u\in V \land v\in V \land (u,v)\in E$ \hfill (consistency) ~~~~~~~~~
\item[]$\mathcal{M}_2$: ~ $\forall u, v, w\in V: (u,v) \in M \land (u,w)\in M \Rightarrow v=w$  \hfill (matching condition)  ~~~~~~~~~
\item[]$\mathcal{M}_3$: ~ $\forall (u,v)\in E, \exists w\in V: (u,w)\in M \lor (v,w)\in M$  \hfill (maximality) ~~~~~~~~~\end{itemize}
\end{definition}

For the sake of simplicity, we assume that if any node $u$ having the value $v$ in its $\beta_{u}$ variable such that $v \not\in V$ or $v\not\in N(u)$ then $u$ understands this value as null ($\bot$).

Algorithm $\mathcal{A}_{1}$ has the three rules described in the following.
If a node $u$ points to null, while one of its neighbors points to $u$, then $u$ accepts the proposition, meaning $u$ points back to this neighbor (\emph{Marriage} rule).  
If a node $u$ points to one of its neighbors while this neighbor is pointing to a third node, then $u$ abandons, meaning $u$ resets its pointer to null (\emph{Abandonment} rule).
If a node $u$ points to null, while none of its neighbors points to $u$, then $u$ searches for a neighbor pointing to null. If such a neighbor $v$ exists, then $u$ points to it (\emph{Seduction} rule). This seduction can lead to either a marriage between $u$ and $v$, if $v$ chooses to point back to $u$ ($v$ will then execute the \emph{Marriage} rule), or to an abandonment if $v$ finally decides to get married to another node than $u$ ($u$ will then execute the \emph{Abandonment} rule). 

We define the probabilistic function $choose(X)$ that uniformly chooses an element in a finite set $X$.

\begin{algorithm}
The process $u$ makes a move according to one of the following rules: 
\begin{itemize}
\setlength{\itemsep}{1mm}
\setlength{\parskip}{1mm}
\item[]\textbf{(Marriage) }  $ (\beta_u=\bot) \land (\exists v \in N(u):\beta_{v}= u) \to \beta_u:= v$
\item[]\textbf{(Abandonment)} $ (\exists v \in N(u): \beta_u= v \land  \beta_v\neq u \land \beta_v\neq \bot) \to \beta_u:= \bot$
\item[]\textbf{(Seduction)} $ (\beta_u=\bot) \land ( \forall v\in N(u):\beta_v\neq u ) \land (\exists v \in N(u): \beta_{v}= \bot) 
	\to$\\ \hspace*{1.5cm} $\textrm{if }choose(\{0,1\})=1 \textrm{ then }\beta_u:=choose(\{v\in N(u): \beta_v=\bot \}) \textrm{ else } \beta_u:=\bot$
\end{itemize}
\end{algorithm}

The node that $u$ chooses to get married with in the marriage rule is not specified, since this choice has no bearing upon  the correctness nor the complexity of the algorithm. 

The proof of this algorithm is based on a potential function. To define this function, we first need to define notions of a \emph{Single node}, a \emph{good edge} and an \emph{almost good edge}.

Let $\mathcal{C}$ be the set of all possible configurations of the algorithm. Let $C\in\mathcal{C}$ be a configuration. 
A process pointing to null and having no neighbor pointing to it is called a \emph{Single node}. Then we define the predicate: 
\emph{Single(u)} $\equiv [\beta_u=\bot \land ( \forall v\in N(u):\beta_v\neq u )]$. 
Moreover, we define the set ${\cal S}(C)$ as the set of Single nodes in $C$. 
Note that if a Single  node is activable then at least one of its neighbors points to null.
We define the two following families of edges:
\begin{itemize}
\itemshort
\item a \emph{good edge} is an edge $(u,v)$ where $\beta_u=v$ and $\beta_v=u$
\item an \emph{almost good edge} is an edge $(u,v)$ where $\beta_u=v$ and $\beta_v=\bot$. Such a node $v$ is called an \emph{Indecisive} process.
\end{itemize}

\begin{figure}[h!] \label{fig_good_almost}
\centering
\scalebox{1}{\includegraphics{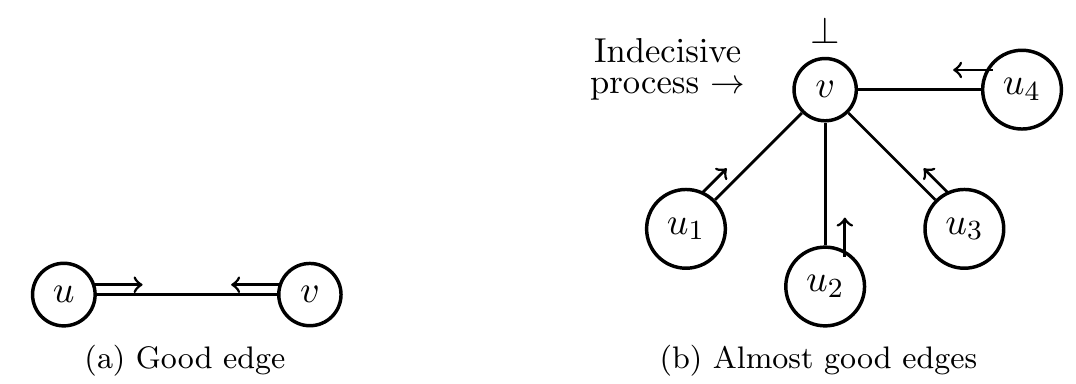}}
%
%
%
%
\caption{Examples of a good and  an almost good edges. Variable $\beta_{.}$ is represented by  the arrows.}

\end{figure}

A process in a good edge cannot ever be activable.
Since every process has only one pointer, two good edges cannot be adjacent.
Therefore, there cannot be more than $n/2$ good edges.
The activation of an Indecisive process necessarily produces  a good edge.
An Indecisive process can belong to many almost good edges. 
So there cannot be more than $n-1$ almost good edge.

We define the potential function $f:\mathcal{C}\to \mathbb{N\times N}$ by
$f(C)=(g,a)$ where $g$ is the number of good edges in $C$
and $a$ the number of almost good edges in $C$.
We recall the lexicographic order on $\mathbb{N\times N}$ :
$(a,b)\preceq (a',b')$ if one of the following conditions holds: \\
\hspace*{4cm} (1) $a < a'$ \hfill
(2) $a=a'$ and $b \leq b'$\hspace*{4cm}

\begin{lemma}\label{lem:1}
Let ${\cal E} = C_0, A_0, \ldots, C_i, A_i, C_{i+1}, \ldots$ be an execution of ${\cal A}_1$. We have: \\
\centerline{$\forall i \in \mathbb{N}, f(C_i) \preceq f(C_{i+1})$}
\end{lemma}
This establishes that $f$  is a  potential function.

\begin{lemma} \label{lem:proba}
Let ${\cal E} = C_0, A_0, \ldots, C_i, A_i, C_{i+1}, \ldots$ be an execution of ${\cal A}_1$.  For all $ i \geq 0 $,    if $A_i$ contains  a move of a node that is Single in $C_i$, then $f(C_i) \prec f(C_{i+1})$  with probability greater than $\frac{1}{4}$. More formally:
$\forall i \in \mathbb{N}, Pr[f(C_i)\prec f(C_{i+1})| \exists u\in {\cal S}(C_i) : u \textrm{ is activated in }C_i]\geq \frac{1}{4}$
\end{lemma}

This means that if the daemon activates a process of ${\cal S}(C)$ then with probability at least $1/4$, 
the potential function $f$ strictly increases.

\begin{lemma} \label{lem:moves}
In any execution containing $n+1$ moves, at least one Indecisive or Single node is activated. 
\end{lemma}

This establishes a upper bound of the number of moves between two  activations of  Indecisive or Single nodes.  From the combination of Lemmas  \ref{lem:1},~\ref{lem:proba} and~\ref{lem:moves}, we obtain Theorem~\ref{algoMatching}.

\begin{theorem}\label{algoMatching}
Under the adversarial distributed daemon and with the guarded-rule atomicity, the matching algorithm ${\cal A}_1$ is self-stabilizing and silent for the specification $\mathcal{M}$ and it reaches a stable configuration in $O(n^3)$ expected moves.
\end{theorem}

We just give an expected upper bound on the convergence time. In the following, we give an upper bound that hold with high probability.

\begin{theorem}\label{th:main}
Let $\varepsilon'>0$, and take $k \geq \text{max}\{4 (n+1)^3, -32 (n+1) \ln\,\varepsilon'\}$. Then, after $k$ moves, the algorithm has converged with probability greater  than $1-\varepsilon'$.
\end{theorem}

\begin{proof}
First, we recall Hoeffding's inequality applied to the identically independent distributed Bernoulli random variables $X_1$, $X_2$, \dots $X_\ell$ which take value $1$ with success probability $1/4$.  Let $S_\ell$ be the random variable such that $S_\ell=\sum_{i=1}^{\ell} X_i$ (corresponding to the number of success during $\ell$ trials). From  Hoeffding's inequality, we have, for some $\varepsilon >0$,
$$ Pr[ S_\ell \leq (1/4-\varepsilon) \ell]  \leq \exp({-2 \varepsilon^2\ell}) $$

Second, we prove the probability  not to reach a stable configuration in at most 
$k$ moves.
Let $C_k$ be the event \emph{``the algorithm has converged in $k$ moves''}. 
Let $F_k$ be a random variable of the number of times the potential function $f$ has increased during the first $k$ moves of the execution. 
Note that $F_k\geq \frac{n^2}2$ means that the  algorithm has converged in at most $k$ moves. We focus on computing an upper bound of the probability of $\neg C_k$. Since $f$ starts from $(0,0)$ in the worst case and is incremented with probability $\frac14$ every $n+1$ moves:
\begin{eqnarray}\label{eq:12}
Pr[\neg C_k] & \leq & Pr \left (F_k  < \frac{n^2}2 \right )  
\end{eqnarray}

From Lemma~\ref{lem:moves}, in any execution containing $n+1$ moves, at least one Indecisive or Single node is activated.  So the number of times that one Indecisive or Single node is activated is at least $\lfloor \frac{k}{n+1}\rfloor =\ell$.
From Lemma \ref{lem:proba}, when a Single or an Indecisive node is actived, $f$ strictly increases with probability greater than $1/4$.   So, this activation can be viewed as a Bernoulli distribution which takes value $1$ with success probability at most $1/4$.


Let $Y_{i}$ be a random variable of the number of times the potential function $f$ has increased after $i$ activations of Single or Indecisive nodes. So Equation~\eqref{eq:12} can be rewritten as:
\begin{eqnarray}
Pr[\neg C_k] & \leq & Pr \left ( Y_\ell  < \frac{n^2}2 \right)  \text{ with  }   \ell=\left \lfloor \frac{k}{n+1}\right \rfloor \label{eq:proof:1}\\
		&\leq & \sum_{i=0}^{\frac{n^2}2-1} {\ell \choose i} \left(\frac14\right)^i\left(\frac34\right)^{\ell-i}
\end{eqnarray}

This value is the probability that less than $\frac{n^2}2$ independent Bernouilli variables with parameter $\frac14$ yield a positive result in $\ell$ trials. Equation \eqref{eq:proof:1} can be rewritten as:
\begin{eqnarray}
Pr[\neg C_k] &\leq & Pr \left ( S_\ell < \frac{n^2}2 \right) \label{eq:proof:2}
\end{eqnarray}
\noindent Thus we can apply the  Hoeffding's inequality, with $\varepsilon = \left (\frac14 - \frac{n^2-2}{2 \ell} \right )$:
\begin{eqnarray}
Pr \left ( S_\ell \leq \frac{n^2}2-1 \right) 
	& \leq & \exp\left(-2\left(\frac14 - \frac{n^2-2}{2 \ell}\right)^2\ell\right) \label{eq:proof:3}
\end{eqnarray}
\begin{eqnarray}
\text{Let }A = 2\left(\frac14 - \frac{n^2-2}{2 \ell}\right)^2\ell\text{, so we can rewrite Eq. \eqref{eq:proof:2}  and \eqref{eq:proof:3} as: } 
Pr[\neg C_k]  \leq \exp (-A)\label{eq:proof:4} 
\end{eqnarray}

Thus, a sufficient condition for the algorithm to have converged after $k$ moves with probability greater than $1-\varepsilon'$ is that:
$Pr[C_k] \geq 1- \varepsilon'  ~\Leftrightarrow~ Pr[\neg C_k] \leq \varepsilon'$
\begin{eqnarray}
\text{We have: }
\exp\left(-A\right) \leq \varepsilon'  ~ \Leftrightarrow ~ A\geq -ln\,\varepsilon'  \label{eq:proof:5}%
\end{eqnarray}%
\begin{eqnarray}\label{eq:proof:6}%
\text{According to Equations \eqref{eq:proof:4}  and \eqref{eq:proof:5}, we have: }
A\geq -ln\,\varepsilon'  \Rightarrow  Pr[\neg C_k] \leq   \varepsilon'
\end{eqnarray}

\noindent We choose the value $k$ such that $A\geq -ln\,\varepsilon' $.
We set $k=\alpha \frac{(n+1)^3}2$, with $\alpha \geq   \text{max}\left\{8,\frac{-64 ln\,\varepsilon'}{(n+1)^2}\right\}$.
\begin{eqnarray}\label{eq:proof:alpha}%
\text{ By computation (see Appendix), we have : } \alpha \geq   \text{max}\left\{8,\frac{-64 ln\,\varepsilon'}{(n+1)^2}\right\} \Rightarrow  Pr[\neg C_k] \leq   \varepsilon'
\end{eqnarray}

Thus, for $\alpha \geq   \text{max}\left\{8,\frac{-64 ln\,\varepsilon'}{(n+1)^2}\right\} $, we have convergence with probability greater than $1-\varepsilon'$.
Finally, since $k=\alpha \frac{(n+1)^3}2$, for $k \geq \text{max}\{4 (n+1)^3, -32 (n+1) \ln\,\varepsilon'\}$, the algorithm has converged after $k$ steps with probability greater than $1-\varepsilon'$.
\end{proof}

\begin{corollary} \label{cor:end}
For $n\geq 6$, the
algorithm converges after $O(n^3)$ moves with probability greater than
$1-\frac{1}{n}$.
\end{corollary}

\section{Handling the anonymous assumption}

When dealing with matching under anonymous networks, we have to overcome the difficulty that a process has to know if one of  its neighbors points to it. In the marriage rule for example, a node $u$ tests if there exists one of its neighbors $v$ such that  $\beta_{v} = u$. Thus $u$ has to know that the $u$'s value appearing in the equality is himself while $u$ has no identity.  This is a fundamental difficulty that is inherently associated to the specification problem and the matching problem cannot even be specified without additional assumptions.  In this paper, the solution consists by adopting the \textit{link-register} model. Moreover, in this model, we assume local unique names on edges/ports (classically named the \emph{port numbering model} in message-passing systems) as a syntactic tool to designate a specific register as well as its counterpart on the other side of the communication link. 

  In this section, we present a self-stabilizing algorithm that gives names to communication links such that (i) a \emph{link-name} is shared by each extremity of the link and (ii) a node cannot have two distinct incident links with the same \emph{link-name}. At the end of the section, we will see how this algorithm is used to overcome the anonymous difficulty previously presented. 

\subsection{The link-register model}

\textbf{The port numbering function:}  
A process $u \in V$ is linked to some other processes, its neighbors, through some edges. Since the network is anonymous, $u$ cannot distinguish these neighbors using identifiers. However, $u$ only knows a \textit{port} associated to each of them. We introduce then the $\mathcal{P}ort$ set being a set of ports, and the function $p$ associating a port to each incident edge of a node. This formalizes as:
{$p : V \times E \rightarrow \mathcal{P}ort \cup \{ ndef \}$}\\[-0.8em]

The function $p$ only makes sense when applied on a process and one of its incident edges. To simplify notations, when $p$ is applied in other cases it returns $ndef$. \\
\centerline{$p\left(u_1,(u_2,u_3 \right))= ndef \; \textrm{if and only if} \;  u_1 \notin \{u_2,u_3\}$}\\[-0.8em]

This leads to model the network as a graph $G=(V,E,p)$. For instance, in the graph of Figure \ref{fig_local_view}, we have: $p(u, (u,v_{1})) = a, p(v_{1}, (u,v_{1})) = d, p(u, (u,v_{2})) = b, p(v_{2}, (u,v_{2})) = c$. Note that $(u,v_1)$ and $(v_1,u)$ denote the same edge and so 
$p(u, (u,v_{1})) = p(u, (v_{1},u)) = a$

\begin{figure}[h!] 
\begin{center}
%
%
%
\scalebox{1}{\includegraphics{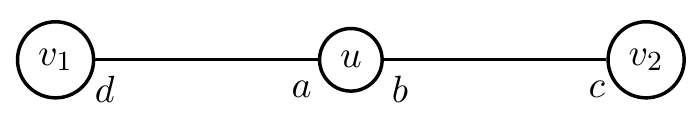}}
\caption{Local view of a node $u$.}
\label{fig_local_view}
\end{center}
\end{figure}

\noindent In this model, all ports for  a given process are distinct:  \\
\centerline{$\forall u{\in} V, \forall (e_1,e_2) {\in} E^2: (p(u,e_1) \neq ndef) {\wedge} (p(u,e_2) \neq ndef) {\wedge} (e_1 \neq e_2)  \Rightarrow (p(u,e_1) \neq p(u,e_2))$}\\
Thus, in Figure \ref{fig_local_view}, $a\neq b$ but $d$ can be equal to $a$, $b$ or $c$.\\

\noindent \textbf{The \emph{port} and \emph{proc} functions:} 
In order to access the necessary information to communicate through the registers, a node needs to know the set of ports of its own incident edges, let $\mathcal{P}(u)$ be this set:\\
\centerline{$\mathcal{P}(u)=\{a \in \mathcal{P}ort : \; \exists e=(u,v) \in E, \;p(u,e)=a\}$}\\
For instance, in Figure \ref{fig_local_view}, $\mathcal{P}_{u}=\{a,b\}$.

We define two functions: $proc$ and $port$. The function $proc$ associates to a node $u$ and one of its ports $a$, the node $v$ reached by $u$ through the port $a$. More formally, we have:
$$proc : V \times \mathcal{P} ort  \rightarrow V \cup \{ndef \}\textrm{ such that } 
proc(u,a)=v \Leftrightarrow 
\begin{cases}
 p(u,(u,v))=a & \mbox{if } a \in \mathcal{P}(u) \\
 v=ndef & \mbox{otherwise}
\end{cases}
$$
For instance, in Figure \ref{fig_local_view}, $proc(u,a)=v_{1}$ and $proc(u,b)=v_{2}$.

The function $port$ is defined accordingly to the function $proc$: given a node $u$ and one of its ports $a$, if $v$ is the node $u$ can reach through port $a$ ($proc(u,a)=v$), then $v$ reached $u$ through $port(u,a)$. More formally, we have:
$$port {:} V {\times} \mathcal{P}ort {\rightarrow} \mathcal{P}ort \cup \{ndef\} \; \textrm{such that }
 port(u,a){=}
   \begin{cases}
        p \left( proc(u,a),(u,proc(u,a)) \right) & \mbox{if } a {\in} \mathcal{P}(u) \\
        ndef  & \mbox{otherwise}
    \end{cases}
$$
For instance, in Figure \ref{fig_local_view}, $port(u,a)=d$, $port(u,b)=c$ and $port(u,c)=ndef$.

Note that the function $proc$ is not used to obtain a node \emph{value} or \emph{identity} but a node \emph{entity}. The function $proc$ is a tool of representation of existing links, allowing a node to denote the node at the other end of one of its communication links.  In the same way, the function $port$ is used to obtain the port \emph{entity} and not the port \emph{value}. 
As a result, we cannot compose or compare results returned by $proc$ or $port$ or do any arithmetic operations on them, this would violate the anonymous assumption.


\subsection{Link-name Algorithm}

The \emph{link-name algorithm} $\mathcal{A}_2$ presented in this section uses the link-register model given in a previous section, with read/write atomicity.

In algorithm $\mathcal{A}_1$, we wrote ``$\,\exists v \in N(u):\beta_v=u\,$'' to specify that a node $v$ is pointing to $u$. The anonymous assumption makes this test impossible ; the link-name algorithm will make it possible.

In algorithm $\mathcal{A}_2$, every node $u$ contains two registers per port $a$ (see Figure \ref{fig_registers}):  $img_{ua}$ and $link_{ua}$. Register $link_{ua}$ is the name $u$ gives to its port $a$. From $u$'s point of view, $link_{ua}$ is the name used by $u$ to designate the node $proc(u,a)$. Register $img_{ua}$ is the name used by $proc(u,a)$ to designate $u$. Once again, from $u$'s point of view, $link_{vb}$ is designated as $link_{proc(u,a)port(u,a)}$ (see \textbf{($R_a$)} rule below).

\begin{figure}[h!]
\begin{center}
\scalebox{1}{\includegraphics{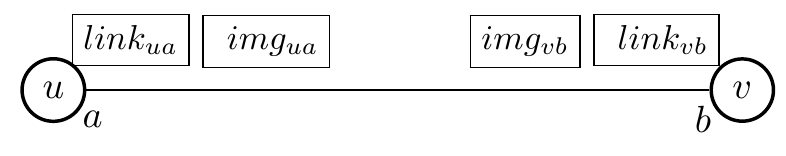}}
\caption{Four registers per edge.}
\label{fig_registers}
\end{center}
\end{figure}

The algorithm ensures that eventually  every two different processes $u$ and $v$ will agree on  names  they use to call each other. More formally, we require the following specification $\mathcal{LN}$ for $\mathcal{A}_2$:

\begin{definition}[Specification $\mathcal{LN}$]
For a graph $G=(V,E,p)$ we have, $\mathcal{LN} = \mathcal{LN}_1 \land \mathcal{LN}_2$ with:
\begin{itemize}
\item[] $\mathcal{LN}_1$: $\forall u \in V, \forall i \in \{ 1, \ldots, |\mathcal{P}(u)|\}, \exists a \in \mathcal{P}(u): \; link_{ua}=i$
\item[] $\mathcal{LN}_2$:
$\forall e=(u,v) \in E: \;  \left(p\left( u,e \right)=a\right) \wedge \left(p\left( v,e \right)=b\right) \Rightarrow (img_{ua}=link_{vb}) \wedge \left( img_{vb} = link_{ua} \right)$
\end{itemize}
\end{definition}

In algorithm $\mathcal{A}_2$, a process $u$ checks whether every port has a unique name taken between 1 and $|\mathcal{P}(u)|$. If not, $u$ renames them all (rule $R_0$). Per port $a$,  process $u$ checks whether the name used by $proc(u,a)$ to designate $u$ is equal to $img_{ua}$ (rule $R_a$).

\begin{algorithm} The link-name algorithm $\mathcal{A}_2$ for node $u$:
\begin{description}
\item[($R_0$)] $\neg (\forall i \in \{ 1, \ldots, |\mathcal{P}(u)| \}, \exists a \in \mathcal{P}(u): \; link_{ua}=i) \rightarrow \text{\emph{Rename all} }\; link_{ua}\; \mbox{\emph{from 1 to}}\; |\mathcal{P}(u)|$
\item[For every $a \in \mathcal{P}(u)$:]~
\begin{itemize}
\item[\textbf{($R_a$)}]  $img_{ua} \neq link_{proc(u,a)port(u,a)}\rightarrow img_{ua}:=link_{proc(u,a)port(u,a)}$
\end{itemize}
\end{description}
\end{algorithm}

Algorithm $\mathcal{A}_2$ satisfies the following theorem:

\begin{theorem}  \label{th:2} Under the adversarial distributed daemon and with the read/write atomicity, the link-name algorithm $\mathcal{A}_2$ is self-stabilizing and silent for the specification $\mathcal{LN}$, and it reaches a stable configuration in $\mathcal{O}(m)$ moves.
\end{theorem}

\subsection{Why is the link-register model not sufficient?}

The link-register model allows to locally distinguish the links incident to a node.  Despite of this link naming, the impossibility result proved by Manne et al.  [19]  still holds. 
In other words, there exists no deterministic self-stabilizing algorithm to build the maximum matching under the synchronous daemon even with a link-register model. So the link-register model allows to overcome the anonymous difficulty that is a node cannot know if one of its neighbors points to itself. However it does not overcome the impossibility in an anonymous network to find some \emph{deterministic} solution for the maximal matching problem.  
In particular, the Manne \emph{et al.} algorithm \cite{ManneMPT09} does not solve the anonymous maximal matching problem even if we assume an underlying link-register model.


\subsection{Rewriting}

In this section, we give a systematic way to rewrite the matching algorithm $\mathcal{A}_1$ using registers of $\mathcal{A}_{2}$ in order to avoid $\mathcal{A}_{1}$'s instructions that violate the anonymous assumption. Algorithm $\mathcal{A}_{1}$ has two kinds of such instructions: the one that writes a non null value in a variable $\beta$ (\emph{e.g.} $\beta_{u} := v$) and the other that searches for a specific non null value in a variable $\beta$ (\emph{e.g.} $\beta_{u} \neq v$). 
For example we would like to see the Marriage rule 
$(\beta_u =\bot) \wedge (\exists v \in N(u): \beta_v=u) \to \beta_u := v$
rewritten as:
$(\beta_u = \bot) \wedge( \exists a \in \mathcal{P}(u): \beta_{proc(u,a)}=img_{ua}) \to \beta_u := link_{ua}$

We give above the generic rules that permit such a rewriting. Note that these generic rules are purely syntactical, \emph{i.e.} we replace some character sequences by some other. In the following, $u$ denotes the node executing the algorithm and $v$ another node ($u\neq v$).
\begin{itemize}
\setlength{\itemsep}{2pt}
\setlength{\parskip}{2pt}
\item The set of neighbor's identifiers $N(u)$ is rewritten as the set of ports $\mathcal{P}(u)$.


\item If $v$ appears 
\begin{enumerate}
\itemshort
\item  in a quantifier, then $u$  manipulates the port $a$ that links it to $v$. So the expression ``$\exists v \in N(u)$''  is replaced by  ``$\exists a \in \mathcal{P}(u)$'' and  the expression ``$\forall v \in N(u)$''  is replaced by ``$\forall a \in \mathcal{P}(u)$'' ;
\item as a subscript of a variable (as in $\beta_v$), then $v$ is replaced by $proc(u,a)$ since in this case $v$ indicates the\emph{ owner of the variable} ;
\item otherwise (as in $\beta_u=v$), $v$ is replaced by $link_{ua}$ since in this case $v$ indicates \emph{the node itself} and so the name used by $u$ to designate node $v$ is needed.\\
\emph{Applying the two previous rules (2. and 3.), we obtain: \\
``$\beta_u = v$'' is rewritten as ``$\beta_u = link_{ua}$'' and  $``\beta_v=\bot''$ is rewritten as ``$\beta_{proc(u,a)} = \bot$''}

\end{enumerate}

\item If $u$ (the node executing the algorithm) appears 
\begin{enumerate}
\itemshort
\item as a subscript of a variable (as in $\beta_u$) then no rewriting is needed,
\item otherwise (as in $\beta_v=u$), $u$ is replaced by $img_{ua}$ since in this case $u$ indicates \emph{the node itself} appearing in the variable of the $u$'s neighbor $v=proc(u,a)$.
\item[]  \hfill \emph{Applying the previous rule, we obtain: ``$\beta_v = u$'' is rewritten as ``$\beta_{proc(u,a)} = img_{ua}$'' ~ }

\end{enumerate}
\end{itemize}

\noindent We give above the algorithm $\mathcal{A}_1$ fully rewritten using these rules:

\begin{algorithm}
~
\begin{itemize}
\setlength{\itemsep}{1mm}
\setlength{\parskip}{1mm}
\item[]\textbf{(Marriage) }  $ (\beta_u=\bot) \land (\exists a \in \mathcal{P}(u):\beta_{proc(u,a)}=  img_{ua}) \to \beta_u:=  link_{ua}$

\item[]\textbf{(Abandonment)} $ (\exists a {\in} \mathcal{P}(u): \beta_u{=} link_{ua} \land  \beta_{proc(u,a)}{\neq}  img_{ua} \land \beta_{proc(u,a)}{\neq} \bot) \to \beta_u:= \bot$

\item[]\textbf{(Seduction)} $ (\beta_u{=}\bot) \land (\forall a {\in} \mathcal{P}(u):\beta_{proc(u,a)}\neq  img_{ua} ) \land (\exists a \in \mathcal{P}(u): \beta_{proc(u,a)}= \bot) 
	\to$\\ 
	\hspace*{2.6cm} $\textrm{if }choose(\{0,1\})=1 \textrm{ then }\beta_u:=choose(\{ a \in \mathcal{P}(u): \beta_{proc(u,a)}=\bot \})$\\
	\hspace*{2.6cm} $\textrm{else } \beta_u:=\bot$
\end{itemize}
\end{algorithm}

\noindent In the following, we will denote by  $\mathcal{RA}_1$, the algorithm $\mathcal{A}_1$ rewritten with the rules above. 

Having defined algorithms $\mathcal{RA}_1$ and $\mathcal{A}_2$, we would like to compose them to give a unified self-stabilizing algorithm.  However, this is not doable in a straightforward way. Indeed, the two algorithms use different communication and atomicity models: algorithm $\mathcal{RA}_1$ assumes the state model with the guarded rule atomicity, while $\mathcal{A}_2$ assumes the link-register model with the read/write atomicity.  For this composition, we keep both models. So $\mathcal{RA}_1$ and $\mathcal{A}_2$ are executed in the same execution, under these two different models. 

We cannot directly apply the composition result of Dolev et al. \cite{Dolev} since authors assume the same model for their composition. However, we can use similar arguments: 
\begin{enumerate}
\item $\mathcal{A}_2$ neither reads nor writes in variables of $\mathcal{RA}_1$ while $\mathcal{RA}_1$ only reads in registers of~$\mathcal{A}_2$. 
\item  $\mathcal{A}_2$ stabilizes independently of $\mathcal{RA}_1$. 
\end{enumerate}
Concerning $\mathcal{RA}_1$, $\mathcal{RA}_1$ has been proved under the state communication model while it  uses \emph{registers} from $\mathcal{A}_2$. However we can notice that $\mathcal{A}_2$ is silent thus the value of these registers will eventually not change. Furthermore, a node does not read registers of its neighbors but only its own registers. Thus they can be viewed as internal variables (since they are not used to communicate between neighbors). Thus $\mathcal{RA}_1$ only uses local and internal variables so the proof that has been done for $\mathcal{A}_1$ is still valid for $\mathcal{RA}_1$.

Thus once $\mathcal{A}_2$ is stabilized  and reaches a stable configuration, $\mathcal{RA}_1$ eventually stabilizes under the state model and the guarded rule atomicity.

 \section{Conclusion}
 
We presented a self-stabilizing algorithm for the construction of a maximal matching. This algorithm assumes the state model and runs in a anonymous network and under the adversarial distributed daemon. It is a probabilistic algorithm that converges in $O(n^3)$ moves with high probability. 

We then present the \emph{pointing impossibility} that is the impossibility for a node, in an anonymous network assuming the state model, to know whether or not one of its neighbors points to it.  We overcome this by using the link-register model. So, we first give a detailed formalization of this model. Second, we gave the \emph{link-name} algorithm, that allow two nodes sharing an edge to keep each other updated about the name they chose for the shared edge. We finally saw that in an anonymous network, assuming the link-register model 
the pointing impossibility result does not hold anymore. 
As a perspective, we would like to analyze the maximal matching algorithm under the link-register model with the read/write atomicity instead of the state model.


\newpage

\appendix
\section{Proofs of Section \ref{sec:matching}}

{\bf Proof of Lemma \ref{lem:1}:} 
A process that belongs to a good edge will never be activable, thus a good edge cannot be destroyed. 

The only way to reduce the number of almost good edges is to apply the marriage rule 
on the Indecisive process. Indeed, let $e = (u,v)$ be an almost good edge, where $v$ is the Indecisive process of $e$.
In edge $e$, the only activable process is $v$ and this node can only execute a marriage move.
Note that $v$ can belong to several almost good edges. Thus the execution of the marriage rule by $v$ destroys all the almost good edges incident to $v$ and creates one new good edge.
To conclude, in every case $f$ cannot decrease. \qed\\\\

{\bf Proof of Lemma \ref{lem:proba}:} 
Let $u$ be a Single node that is activated in $C_i$. Let $F(u)=\{v\in N(u): \beta_v=\bot \}$ and $d$ be its cardinality.
Since $u$ is activated in $C_i$, $F(u)$ is not an empty set and  $d\geq 1$.
 
\begin{figure}[h!] 
\begin{center}
%
%
\scalebox{1}{\includegraphics{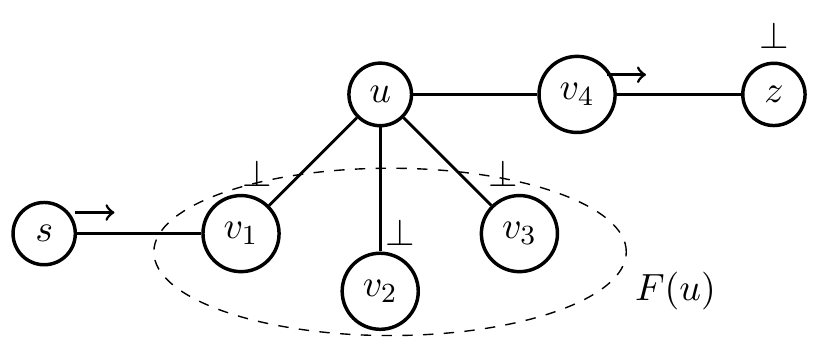}}
\end{center}
 \caption{Example of a configuration where $u$ is a Single   node and $F(u)=\{v_1,v_2,v_3\} $. } 
 \label{fig_proba_calcul} 
 \end{figure}

Let $v\in F(u)$. Note that $u$ seduces a neighbor $v$  with a probability $\frac{1}{2d}$. Indeed, process $u$  seduces one of its neighbors with probability $\frac{1}{2}$. Moreover, if $u$ decides to seduce (\emph{i.e.} to change its local variable $\beta_u$), then it  seduces $v$ with probability   $\frac{1}{d}$. Let us assume that $u$ seduces $v$, then there are  three cases (see Figure~\ref{fig_proba_calcul}):
\begin{description}
\setlength{\itemsep}{1mm}
\setlength{\parskip}{1mm}
\item [Case 1) $v\not\in {\cal S}(C)$ and $v$ is activated in $C_i$:]  ($v_1 \not\in {\cal S}(C)$ in Fig. \ref{fig_proba_calcul}) In this case, there is a process  $s$ such that $\beta_{s}=v$ 
and $v$ has to apply the marriage rule. This action creates a new good edge containing $v$.
\item [Case 2) $v\in {\cal S}(C)$ and $v$ is activated  in $C_i$:] ($v_2 \in {\cal S}(C)$ and $v_3 \in {\cal S}(C)$ in Fig. \ref{fig_proba_calcul}) In this case, process $v$ applies the seduction rule and still points to null with probability $\frac{1}{2}$.
Thus,  edge $(u,v)$ becomes an almost good edge with probability~$\frac{1}{2}$.
\item [Case 3) $v$ is not activated  in $C_i$:] In this case $v$ still points  to null with probability 1. Thus,    edge $(u,v)$ becomes an almost good edge.
\end{description}

\noindent Let $Pr_{uv}$ be the probability to create:
\begin{itemize}
\setlength{\parskip}{-2mm}
\item  a new almost good edge between $u$ and $v$ such that $v$ is the Indecisive process of edge $(u,v)$  (we will call this an \emph{almostGood-creation}),
\itemshort
\item  or a new good edge containing $v$  (we will call this a \emph{good-creation}),
\end{itemize}
during $C_i \to C_{i+1}$ given that $u$ is a Single node that is activated in $C_i$ and $v\in F(u)$ and $|F(u)| = d$. 
\\[-0.5em]

From the previous case study, we can deduce that
 in case 1, a \emph{good-creation} occurs with probability $\frac{1}{2d}\times 1$, 
 in case 2, an \emph{almostGood-creation} occurs with probability $\frac{1}{2d}\times \frac12$ and
 in case 3, an \emph{almostGood-creation} occurs with probability $\frac{1}{2d}\times 1$. 
Thus $Pr_{uv} \geq \min \{ \frac{1}{2d}, \frac{1}{4d}, \frac{1}{2d} \} \geq \frac{1}{4d}$.\\[-0.5em]


\noindent Moreover, we define $Pr_u$ as the probability for $u$ to create: 
\begin{itemize}
\setlength{\parskip}{-2mm}
\item a new almost good edge between $u$ and one of its neighbors $v$  such that $v$ is the Indecisive process of edge $(u,v)$
\itemshort
\item  or a new good edge containing one of its neighbor, 
\end{itemize}
during $C_i \to C_{i+1}$ given that $u$ is a Single node that is activated in $C_i$. \\

\noindent We have:
$Pr _u =\sum_{v\in F(u)}  Pr_{uv} \geq \frac{1}{4}$

\noindent Thus :
$Pr[f(C_i)\prec f(C_{i+1})| \exists u\in {\cal S}(C_i) \textrm{ activated in }C_i]\geq \frac{1}{4}$

\qed\\\\
%

{\bf Proof of Lemma  \ref{lem:moves}:}
We are going to exhibit an as long as possible execution $\mathcal{E}$ without activating any Indecisive or Single node. Note that the seduction rule can only be executed by Single nodes and that the marriage rule can only be executed by Indecisive nodes. Thus, the only possible rule in $\mathcal{E}$ is the Abandonment rule. If this rule is executed by some process $u$, then $u$ becomes Single and all Indecisive, respectively Single, neighbors of $u$ will remain so. 
Thus, we can apply this rule at most $n$ times. 
\qed\\\\

{\bf Proof of Theorem \ref{algoMatching}:}
First, we prove that a stable configuration satisfies the
specification $\mathcal{M}$.  In a stable configuration $C$, we can apply neither the
abandonment nor the marriage rules.  
Then for all node $u$ in $V$  we have one of the two following conditions:
$(\beta_u=\bot \land \forall v \in N(u): \beta_v \neq u)$ or $(\exists v \in N(u) : \beta_u = v \land \beta_v = u)$.
Note that in condition 2, $\beta_v$ cannot be $\bot$ otherwise condition 1 does not hold for $v$.
Thus the set $M = \{(u,v):\beta_u=v \land \beta_v=u \}$ is such that $(u,v)\in E$ and so $\mathcal{M}_1$ holds.  


Let us prove that in a stable coation, $M$ is a maximal
matching.  Since every process has only one pointer, two good edges
cannot be adjacent.  Thus $\mathcal{M}_2$ holds and $M$ is a
matching. Since we cannot apply the seduction rule, for every node $u$ such
that $\beta_u = \bot$,  every neighbor $v$ of $u$ is such that
$\beta_v \neq \bot$ and thus $v$ belongs to  an edge in $M$. Thus
$\mathcal{M}_3$ holds  and $M$ is maximal.




Now, let us prove that we reach a stable configuration in $O(n^3)$ expected moves.  Let $X$ be a random variable of the number of moves needed to increase function $f$.  According to Lemma \ref{lem:moves}, each sequence of $n+1$ moves, contains at least one move from an Indecisive or Single node. Let $m$ be such a move. If $m$ is a move of an Indecisive node, then $m$ is the execution of the Marriage rule, and thus $f$ increases by one during this move. If $m$ is a move of a Single node, then according to Lemma 
\ref{lem:proba}, $f$ strictly increases with probability greater than $1/4$ during this move. Thus in both cases, $f$ strictly increases with probability greater than $1/4$ during $m$.  If it fails, then there can have at most $n$  additional moves before the activation of an Indecisive or Single node.  And there can be at most $n-1$ simultaneous  moves when an Indecisive or Single node is actived. We have a sequence of Bernoulli trials, each with a probability greater than $1/4$ of success.  So, we have $E[X] \leq 4\cdot 2\cdot  n$.  By definition,  function $f$ has  $O(n^2)$ possible values. Therefore, the expected time for $f$ to reach its maximal possible value is $O(n^3)$.  In conclusion, algorithm ${\cal A}_1$  reaches a stable configuration in $O(n^3)$ expected moves. \qed\\\\

{\bf Proof of Equation~\eqref{eq:proof:alpha}:}
Let $k=\alpha \frac{(n+1)^3}2$, with $\alpha\geq 8$, and so $\ell \geq \frac{\alpha (n+1)^2}{2}$. 

We can get: 
$$\frac{2-n^2}{2\ell} = \frac{3+2n}{\alpha(n+1)^2}-\frac1\alpha$$
$$
\text{Thus: }A= 
{\alpha (n+1)^2} \left (\frac14-\frac1\alpha+\frac{3+2n}{\alpha(n+1)^2}\right )^2
$$
$$
\text{Since }\alpha\geq 8: 
\frac14-\frac1\alpha\geq \frac18 \text{ and so, }\frac14-\frac1\alpha+\frac{3+2n}{\alpha(n+1)^2} > 
\frac18 \text{ and }
\left (\frac14-\frac1\alpha+\frac{3+2n}{\alpha(n+1)^2}\right )^2 > \frac1{64}
$$
$$
\text{So: }{\alpha (n+1)^2} \left (\frac14-\frac1\alpha+\frac{3+2n}{\alpha(n+1)^2}\right )^2 > \frac{\alpha (n+1)^2}{64} 
\Rightarrow A \geq  \frac{\alpha (n+1)^2}{64} 
$$
$$
\text{Taken }\alpha \geq \frac{-64 ln\,\varepsilon'}{(n+1)^2}\text{, we have: } 
\frac{\alpha (n+1)^2}{64} \geq -ln\,\varepsilon'    \Rightarrow A \geq -ln\,\varepsilon' 
$$
$$
\text{According to Equation  \eqref{eq:proof:6},  we have: } 
\alpha \geq   \text{max}\left\{8,\frac{-64 ln\,\varepsilon'}{(n+1)^2}\right\} \Rightarrow  Pr[\neg C_k] \leq   \varepsilon'
$$ \qed\\\\

{\bf Proof of Corollary \ref{cor:end}:}
First, we notice that the  function $f(x)= 4(x+1)^3 - 32(x+1) \log{x}$ is an increasing function for    $x \geq 6$. In fact $f^{\prime}(x)=12(x+1)^2 - 32 (1+\frac{1}{x}-\log{x})$ and  $f^{\prime}(x)$ has positive
values for  $x\geq 6$. This implies that when $x\geq 6$, we have $4(x+1)^3 \geq 32(x+1) \log{x}$ 

Second, we focus on $n\geq 6$. Let $\varepsilon^{\prime} $ be a real such that $\varepsilon^{\prime} <1/n$. From Theorem~\ref{th:main},  the algorithm has converged after $k$ moves with probability greater  than $1-\varepsilon'$ where $k \geq \text{max}\{4 (n+1)^3, -32 (n+1) \ln\,\varepsilon'\}$.
Since $n\geq 6$, we obtain that $k \geq 4(n+1)^3 \geq -32(n+1) \log{\frac1n}$.  So,  the algorithm has converged after $ 4(n+1)^3$ moves with probability greater  than $1-\frac1n$. \qed

\section{Proofs of Theorem \ref{th:2}}

\begin{lemma} \label{lemma1 naming}
Under the adversarial distributed daemon and with the read/write atomicity, any execution of the link-name algorithm ${\cal A}_2$ reaches a stable configuration in at most $20m$ moves.
\end{lemma}

\begin{proof}
We start by giving the complexity of algorithm $\mathcal{A}_2$ in term of moves. First, we focus on rule $R_0$ for each node $u$. Node $u$ executes rule $R_0$ at most once. 
 During this execution, for every port $a$, $u$ makes at most $2$ moves: one reading-move to check whether $link_{ua}=i$ and at most one writing-move to rename $link_{ua}$. Therefore, the total number of moves for $R_0$ executed by any node in the network is: $\sum_{u \in V} 2|\mathcal{P}(u)|=4m$

Second, we count the move complexity for rule $R_a$ given a node $u$ and a port $a$. 
$R_a$ is executed by $u$ at most twice, since  $link_{proc(u,a)port(u,a)}$ changes its value at most once (due to the rule $R_0$).  In every execution of rule $R_a$, $u$ makes at most $4$ moves: 2 moves to compare the two registers $img_{ua} \neq link_{proc(u,a)port(u,a)}$ and, if the condition holds, 2 moves to assign $link_{proc(u,a)port(u,a)}$ to $img_{ua}$. Therefore, the total number of moves for $R_a$ executed by any node in the network, and for all $a$ is: $ \sum_{u \in V} 2 \cdot 4 \cdot |\mathcal{P}(u)| = 16m$

By summing up the number of moves required by all rules, we obtain: the maximum number of moves in any execution of the algorithm is $20m$ moves.
\end{proof}\\\\

\begin{lemma}\label{lemma2 naming}
In every stable configuration of $\mathcal{A}_2$, the specification $\mathcal{LN}$ holds.
\end{lemma}

\begin{proof}
In a stable configuration, the guard of rule $R_0$ is false for every node. Thus $\mathcal{LN}_0$ holds. For every edge $e=(u,v)$ in the network, neither $u$ nor $v$ can apply their rule associated to edge $e$: $R_{p(u,e)}$ and $R_{p(v,e)}$. This implies that $img_{ ua} = link_{vb}$ and $img_{vb} = link_{ua}$, where $a=p(u,e)$ and $b=p(v,e)$. Thus $\mathcal{LN}_1$ holds.
\end{proof}
~\\~\\
From the combination of Lemmas \ref{lemma1 naming} and \ref{lemma2 naming}, we obtain Theorem \ref{th:2}.

\end{document}